\documentclass{llncs}
\let\qedhere\qed

\makeatletter
\let\claim\@undefined
\let\endclaim\@undefined
\makeatother
\spnewtheorem{claim}[theorem]{Claim}{\bfseries}{}

\usepackage{graphicx}
\usepackage{tikz}
\usepackage[caption=false]{subfig}
\usepackage{amsmath,amssymb}
\usepackage{mathtools}
\usepackage{stmaryrd,xcolor,booktabs,enumitem,microtype}
\usepackage[utf8]{inputenc}
\usepackage[T1]{fontenc}
\usepackage{lmodern}

\let\qedhere\qed

\DeclarePairedDelimiter\abs{\lvert}{\rvert}%
\newcommand{\size}[1]{\abs{#1}}%

\newcommand{\defterm}[1]{\emph{#1}}%

\usepackage{autonum}

\DeclarePairedDelimiter\multiset{\{}{\}}

\newcommand{\todom}[1]{\ignorespaces}

\newcommand{\tbn}{}\def\tbn/{TBN}%
\newcommand{\tbns}{}\def\tbns/{TBNs}%
\newcommand{\anoutsideterm}{}\def\anoutsideterm/{an interface}%
\newcommand{\outsideterm}{}\def\outsideterm/{interface}%
\newcommand{\outsidesterm}{}\def\outsidesterm/{interfaces}%

\newcommand{\flip}[1]{#1^{*}}%

\newcommand{\site}[1]{\ifcase#1\or s\or t\or u\else\@ctrerr\fi}%
\newcommand{\sitetype}[1]{\ifcase#1\or a\or b\or c\else\@ctrerr\fi}%
\newcommand{\monomer}[1]{\mathbf{\ifcase#1\or m\or n\or r \else\@ctrerr\fi}}%
\newcommand{\therm}[1]{\mathcal{\ifcase#1\or T\or V\else\@ctrerr\fi}}%
\newcommand{\config}[1]{\ifcase#1\or \gamma\or \delta\else\@ctrerr\fi}%
\newcommand{\polymer}[1]{\ifcase#1\or P\or Q\else\@ctrerr\fi}%
\newcommand{\stablepolycount}{S}%

\newcommand{\class}[1]{{\textrm{#1}}}%
\newcommand{\ccomplete}[1]{#1\textrm{-complete}}%
\newcommand{\cp}{\class{P}}%
\newcommand{\cnp}{\class{NP}}%
\newcommand{\cpnp}{\cp^\cnp}%
\newcommand{\cparallelnp}{\cp^\cnp_{\parallel}}%
\newcommand{\capx}{\class{APX}}%
\newcommand{\cfapx}[1]{#1\textrm{-}\capx}%

\newcommand{\prob}[1]{\textsf{#1}}%
\newcommand{\psat}{\prob{SAT}}%
\newcommand{\pcount}{\prob{SaturatedConfig}}%
\newcommand{\pfree}{\prob{StablyFree}}%
\newcommand{\ptogether}{\prob{StablyTogether}}%
\newcommand{\stablyq}{\prob{Stably}\,Q}%

\newcommand{\ab}{\allowbreak}
\newcommand{\pindmember}{
  \prob{Max} \prob{Ind} \ab \prob{Set} \ab \prob{Member}}%
\newcommand{\pindnotmember}{
  \prob{Max} \prob{Ind} \ab \prob{Set} \ab \prob{Not} \ab \prob{Member}}%
\newcommand{\pindrmember}{
  \prob{Max} \ab \prob{Ind} \ab \prob{Set} \ab \prob{(Not)} \ab \prob{Member}}%
\newcommand{\pcovermember}{
  \prob{Min} \ab \prob{Vertex} \ab \prob{Cover} \ab \prob{Member}}%
\newcommand{\pcovernotmember}{
  \prob{Min} \ab \prob{Vertex} \ab \prob{Cover} \ab \prob{Not} \ab \prob{Member}}%
\newcommand{\pcoverrmember}{
  \prob{Min} \ab \prob{Vertex} \ab \prob{Cover} \ab \prob{(Not)} \ab \prob{Member}}%

\newcommand{\reducesto}{\leq}%

\newcommand{\pflip}[1]{%
  \makebox[0pt][l]{\ensuremath{\smash{\flip{#1}}}}%
  \phantom{#1}}%
\newcommand{\falign}[1]{\makebox(0,0)[c]{#1}}%
\newcommand{\panel}[4]{%
  \setlength{\fboxsep}{#2}%
  \colorbox{#1}{\hspace{#3}#4\hspace{#3}}}%
\newcommand{\spanel}[1]{\panel{white}{0pt}{0pt}{#1}}%
\newcommand{\U}[1]{\falign{\spanel{\ensuremath{\fsite#1}}}}%
\newcommand{\X}[1]{\falign{\spanel{\pflip{\fsite#1}}}}%
\newcommand{\tpanel}[1]{\panel{white}{0pt}{0.4em}{#1}}%
\newcommand{\ftext}[1]{\falign{\tpanel{#1}}}%
\newcommand{\fitext}[1]{\ftext{\itshape{#1}}}%

\newcommand{\from}{\colon}%

\newcommand{\aftertheoremindent}{}

\begin{document}

\title{Computing properties of stable configurations\\
of thermodynamic binding networks}
\author{
  Keenan Breik\inst{1},
  Chris Thachuk\inst{2},
  Marijn Heule\inst{1},
  David Soloveichik\inst{1}
}
\institute{University of Texas at Austin \and California Institute of Technology}
\maketitle{}

\begin{abstract}
  
The promise of chemical computation lies in controlling systems incompatible with traditional electronic micro-controllers, with applications in synthetic biology and nano-scale manufacturing.
Computation is typically embedded in kinetics---the specific time evolution of a chemical system. 
However, if the desired output is not thermodynamically stable, basic physical chemistry dictates that thermodynamic forces will drive the system toward error throughout the computation. 
The thermodynamic binding network (TBN) model was introduced
to formally study how the thermodynamic equilibrium
can be made consistent with the desired computation,
and it idealizes tradeoffs between configurational entropy and binding.
Here we prove the computational hardness
of natural questions about TBNs
and develop a practical algorithm
for verifying the correctness of constructions
by translating the problem into propositional logic
and solving the resulting formula.
The TBN model together with automated verification tools
will help inform strategies for error reduction in molecular computing,
including the extensively studied models of strand displacement cascades
and algorithmic tile assembly.

\end{abstract}


\section{Introduction}

Similar to digital electronics, advances in engineering of molecular computation have relied on a distinctive set of abstractions and models.
The formalism of algorithmic tile assembly~\cite{doty2012theory} enabled the self-assembly of complex nanostructures from simple parts,
with a molecular computational process
(such as simulating binary counting)
directing component placement%
~\cite{barish2009information,ong2017programmable}.
\todom{DS: Ref~\cite{ong2017programmable} to be replaced with Dave/Damien's amazing iterated circuits paper, when it finally gets published.}
Likewise, DNA strand displacement cascades (formalized as~\cite{phillips2009programming}) made it possible to rationally design molecular reaction pathways,
and this model has been used to engineer a wide range of molecular devices, programmable structures, logic and neural circuits, and dynamical systems~\cite{zhang2011dynamic,srinivas2017enzyme,cherry2018scaling}.
The ultimate applications of molecular computation are in contexts where traditional electronics cannot be used. 
These applications include reprogramming biological cell behaviors
or controlling complex nanoscale assembly processes.
It is also hoped that theories of molecular computing can shed light on ill-understood design principles of natural biological regulatory and development pathways.

The widely studied models of chemical computing such as algorithmic tile assembly and strand displacement cascades 
are essentially kinetic as they describe a desired time evolution of an information processing chemical system.
However, unlike electronic computation, 
chemical computation operates in a Brownian environment subject to powerful thermodynamic driving forces.
If the desired output happens to be a meta-stable configuration, then thermodynamic driving forces will inexorably drive the system toward error.
For example, in tile assembly, thermodynamically favored assemblies that are not the intended self-assembly program execution are a major source of error~\cite{SchWin09,barish2009information}. 
Likewise,
spurious production of signal in most strand displacement systems (called leak) occurs because the thermodynamic equilibrium of a strand displacement cascade favors incorrect over the correct output, or does not discriminate between the two~\cite{thachuk2015leakless}.

The thermodynamic binding network (TBN) model abstracts chemical systems at
the thermodynamic equilibrium~\cite{tbn}.
The model is simple and general due to its two main features: 
(1) abstracting away of ``geometry'', 
and (2) the simplification of thermodynamics to a tradeoff between
the number of separate complexes (configurational entropy)
and the number of bonds formed.
These features of the model make it widely applicable, including for  
understanding the consistency of kinetics and thermodynamics for strand displacement cascades, algorithmic tile assembly, as well as other contexts.
The simplicity of the model makes it amenable to rigorous proofs.

The most basic question is how can we distinguish
(thermodynamically) \emph{stable} configurations---%
that is configurations that are favorable in terms of free energy---%
from unstable ones.
Predicting the thermodynamic equilibrium is often computationally difficult: for example,
determining the lowest energy configurations of Ising models~\cite{cipra2000ising}, predicting secondary structure of nucleic-acids with pseudoknots~\cite{lyngso2000rna}, and predicting protein folding~\cite{hart1997robust} were shown to be NP-hard.
These problems derive their hardness from geometrical constraints.
In contrast the TBN model avoids geometry, and the computational complexity originates in the interplay between the opposing forces of increasing binding and increasing the number of separate complexes.
Note, however, that even without configurational entropy, there are interesting computational problems derived solely from geometry-free binding%
~\cite{jonoska2011stoichiometry}.%

In the TBN model,
the hard part of checking whether a given configuration is stable
is determining whether the system can reconfigure
to increase the number of separate complexes,
thereby increasing configurational entropy,
without reducing the total number of bonds.
We prove that this problem is NP-complete in the worst case.
We strengthen this further:
letting $S$ denote the number of separate complexes in a stable configuration,
we show that computing $S$
is not in $\cfapx{n^{\delta}}$ for any $\delta < 1$ unless $\cp = \cnp$
(that is, $S$ is hard to approximate).

Applications of TBNs in molecular computing
also motivate more complex questions about stable configurations.
In DNA strand displacement cascades,
output is usually represented by the release of a previously bound DNA strand.
The question of whether releasing the output is thermodynamically favorable
corresponds to the problem of deciding whether a given strand is free
in some stable configuration of the TBN model ($\pfree$).
In other contexts,
the output is represented by whether two specific molecules
are bound in the same complex,
and we want to determine whether there is a stable configuration
with that property ($\ptogether$).
By connecting these problems to a graph problem
concerning membership in a minimum vertex cover,
we show that these problems are complete for $\cparallelnp$
(which is $\cp$ with parallel access to an $\cnp$ oracle),
a complexity class of growing importance~%
\cite{kadin1989pnp,fal2009richer}.

Despite these worst-case negative results,
we develop a software package accompanying this paper
that can answer many questions in practice,
with functionality for finding stable configurations,
and finding stable configurations with specific properties
including $\pfree$ and $\ptogether$%
~\cite{solver}.
To compute $S$,
the package computes a non-trivial reduction
to the Boolean satisfiability problem (SAT)
which is then passed to a SAT solver.
We show that an exponential speed up can be achieved in certain cases, compared with a naive solution based on enumerating all configurations, or even the subset of maximally bound (saturated) configurations.
Although $\pfree$ and $\ptogether$
are more complex questions about stable configurations,
we show that these problems can be reduced
to computing $S$ on a transformed TBN.
Our package assists in manipulating and understanding
the behavior of TBNs.

Our ultimate goal is molecular computation in which the thermodynamic equilibrium is consistent with the desired computation, and that it can be reached by an efficient kinetic pathway.
The TBN model permits differentiating thermodynamic and kinetic contributions to the computational power of the chemical system, and by ensuring that both are consistent we can ensure greater fidelity of molecular computing.
Note that although we show that NP-hard problems can be encoded in TBNs, we in no way claim that TBNs by themselves provide an effective physical mechanism for solving NP-hard problems.
Indeed, without a separate kinetic argument as that given by strand displacement or tile assembly, there are no guarantees on the time required to approach thermodynamically favored states.

\section{Model}

TBNs consider chemical systems affected by two (often opposing) thermodynamic driving forces:
the free-energy benefit due to forming additional bonds,
and the penalty for separate molecules joining together (free-energy of association). 
TBNs focus on these two driving forces due to their wide applicability, as well as the existence of systematic ways to amplify their strength relative to other driving forces~\cite{tbn}.
A further aspect of the model that makes it inclusive is that it does not rely on geometric constraints to enforce correct behavior.
Thus, in the TBN model, we think of a monomer
as simply being an unstructured collection of binding sites.

\newcommand{\sitetypeset}{D}%
\newcommand{\sitetypelit}[1]{\sitetype#1}%
Formally,
a \tbn/ is a triple $(\sitetypeset, {\flip{}}, \therm1)$
where
\begin{itemize}
  \item
    $\sitetypeset$ is a finite set,
    which we call the set of \defterm{site type}s.
    These represent specific binding motifs, 
    with the prototypical example of DNA ``domains''---sequences that are designed to bind as a unit. 

  \item
    $\flip{} \from \sitetypeset \to \sitetypeset$
    is an involution (its own inverse)
    with no fixed point
    ($\flip{\sitetypelit1} \neq \sitetypelit1$ for all $\sitetypelit1$).
    This way the \defterm{complement}
    of the site type
    $\flip{\sitetypelit1}$
    is $(\flip{\sitetypelit1})\flip{{}} = \sitetypelit1$.
    The prototypical example is that of DNA,
    where two sequences are complementary
    if they are Watson-Crick complements of each other.%
    \footnote{
      We exclude DNA reverse complement palindromes like AAGAATTCTT,
      which can self-hybridize.}

  \item
    $\therm1$ is a finite multiset of monomer types,
    and a \defterm{monomer type}
    is a finite multiset over $\sitetypeset$.
    This models that
    chemical systems consist of macromolecules
    and that
    macromolecules may have multiple binding sites.
\end{itemize}

We often call $\therm1$ alone a \tbn/
when $\sitetypeset$ and $\flip{}$ are to be inferred.
\newcommand{\extherm}{\therm1_\mathrm{ex}}%
For example,
given the following multiset of monomer types,
\begin{equation}
  \extherm =
  \multiset{
    \multiset{\sitetypelit1, \flip{\sitetypelit2}},
    \multiset{\sitetypelit1, \flip{\sitetypelit2}},
    \multiset{\flip{\sitetypelit1}, \flip{\sitetypelit1}, \sitetypelit2}
  },
\end{equation}
we can infer that $D =\{a, b, \flip{a}, \flip{b}\}$
and that $\flip{}$ maps $a$ to $\flip{a}$, $b$ to $\flip{b}$ and vice versa.
We call an occurrence of a site type a \defterm{site},
and an occurrence of a monomer type a \defterm{monomer}.
So $\extherm$ has
seven sites but just four site types,
and it has three monomers but just two monomer types.
A site (type) is \defterm{limiting}
if its complement occurs at least as often.
So $\sitetypelit1$, $\flip{\sitetypelit1}$, and $\sitetypelit2$
are limiting in $\extherm$.

We use a bold variable like $\monomer1$ to indicate a monomer
and an italic variable like $\site1$ to indicate a site.

\begin{figure}[h]
  \newcommand{\fsite}[1]{\ifcase#1\or a\or b\else\@ctrerr\fi\vphantom{bg}}%
  \newcommand{\fsa}{a}%
  \newcommand{\fsb}{b}%
  \newcommand{\Tsat}{\fitext{saturated}}%
  \newcommand{\Tstab}{\fitext{stable}}%

  \def\svgwidth{\columnwidth}
  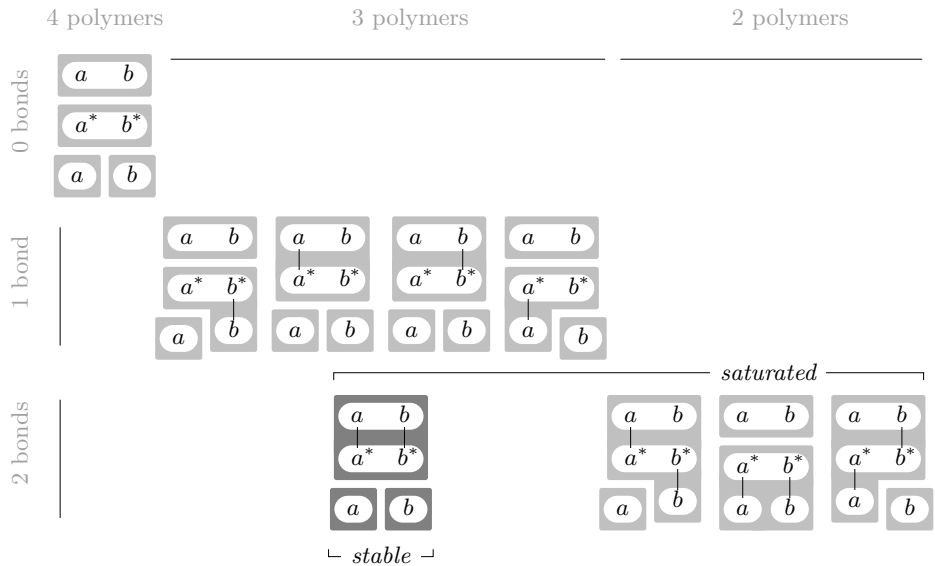

  \caption{
    All nine configurations of the TBN
    $\therm1 = \{\{\fsa\}, \{\fsb\}, \{\fsa, \fsb\}, \{\flip{\fsa}, \flip{\fsb}\}\}$.
    We can infer the set of site types to be
    $\sitetypeset = \{a, b, \flip{a}, \flip{b}\}$,
    and we can infer that $\flip{}$ maps $a$ to $\flip{a}$
    and $b$ to $\flip{b}$ and vice versa.
    An edge between sites indicates that they pair in that configuration.
    A shaded box indicates a polymer.}
  \label{configs}
\end{figure}

As Figure~\ref{configs} illustrates,
a \defterm{configuration} $\config1$ of a \tbn/
is a matching among
its complementary sites.
Two sites \defterm{pair} in $\config1$ if they are matched.
Two monomers \defterm{bind} in $\config1$ if some of their sites pair.
A \defterm{polymer} of $\config1$ is a connected component
with respect to binding.
The configuration $\config1$ is \defterm{saturated}
if the matching is maximal.
$\config1$ is \defterm{stable} if it is saturated
and no saturated configuration has more polymers.
The number of polymers in $\config1$ is $\size{\config1}$.

All else being equal,
a state with more bonds formed is more favorable,
and a state with more separate polymers is more favorable.%
\footnote{Intuitively, the entropic benefit is due to additional microstates,
each describing the three-dimensional position of each polymer,
associated with such a state.
See \cite{tbn} for additional physical intuition.}
Thus ``thermodynamic stability'' is equated with
``maximizing the number of molecular bonds and the number of separate polymers''
subject to some prescribed trade-off between the two.
Although the general case of quantitative trade-off is complex,
\tbns/ take the limiting case
in which maximizing the number of bonds is infinitely preferred.
Importantly, many systems studied in molecular programming
operate in this limit:
in particular, DNA systems with long domains (strong binding sites)
at relatively low concentrations~\cite{thachuk2015leakless}.%
\footnote{
  We can systematically approach this limit in two steps.
  First reduce concentration just enough
  to make entropic forces significant.
  Then make domains long enough that binding dominates.%
}

The TBN model abstracts a volume with one copy each of the given monomers.
To handle multiple copies, each copy must be explicitly included. 
Certain results that hold in the single copy setting
could be generalized to a ``bulk'' setting
with a variable count of each monomer type,
although this is an area for future work (see Section~\ref{conclusion}).

\section{Examples} \label{sec:ex}

The following examples
illustrate the kinds of analysis the TBN model allows,
motivated by the verification of different chemical systems.

{\begin{figure}[t]
  \newcommand{\fsite}[1]{\ifcase#1%
  \or a\or b\or c\or d\or e\or f\or g%
  \else\@ctrerr\fi\vphantom{bg}}%

  \subfloat[]{
    \def\svgwidth{0.46\textwidth}
    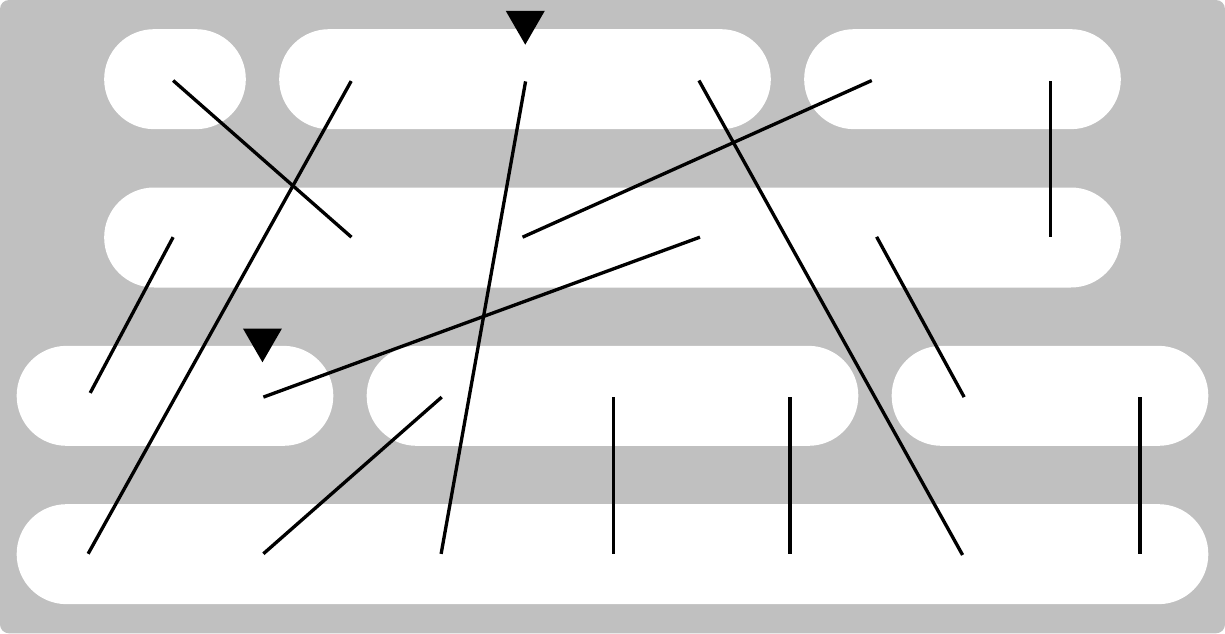
  }
  \hspace*{\fill}
  \subfloat[]{
    \def\svgwidth{0.46\textwidth}
    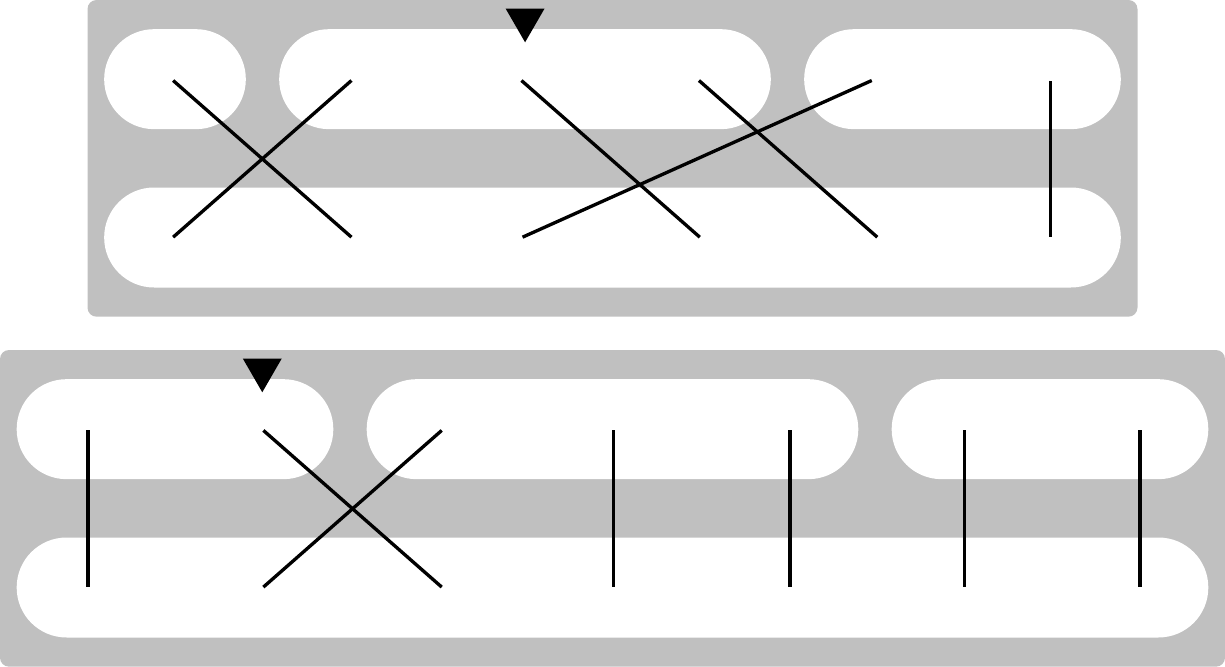
  }

  \caption{
    (a) A stable configuration of a TBN $\therm1_a$
    with a single polymer.
    Every stable configuration of $\therm1_a$
    has a single polymer.
    (b) A stable configuration of another TBN $\therm1_b$.
    The only difference between $\therm1_a$ and $\therm1_b$
    is the two sites marked by triangles swapping.
    With this subtle difference,
    $\therm1_b$ can fall apart into two polymers.
    }
  \label{fig:split}
\end{figure}
}

A core part of chemical and molecular machinery
is assembling stable molecular structures.
For example, biology relies
on molecular structures like motors and capsids.
Similarly, in engineering,
molecular self-assembly techniques design building blocks
that accrete into large structures.
A basic question of self-assembly
is will a desired structure hold together and how can we know?

Distinguishing stable from unstable assemblies can be challenging. 
Figure~\ref{fig:split} shows two TBNs,
each consisting of nearly the same collection of monomers,
but there is a sharp difference in their behavior.
One may fall apart into two polymers in a stable configuration,
while every stable configuration of the other holds together
as a single polymer.
Can we tell the two cases apart
without enumerating exponentially many configurations?

As a more complex example of molecular computation,
consider molecular logic circuits proposed in prior work%
~\cite{chalk2018thermodynamically}.
Figure~\ref{fig:circuit} shows a Boolean sorting network
of 10 logic gates implemented with 97 monomers.
The inputs are encoded in an ``input'' monomer.
Each output is represented by which of two ``output'' monomers
(representing 0 and 1 respectively)
is bound in the same polymer as the input.
Although the details of the construction are beyond the scope of this paper,
intuitively each logic gate corresponds to a set of monomers,
with a monomer for each entry of the truth table
(for example, $\{1_a, 0_c, 1_i^*, 1_j^*\}$ corresponding to the OR
with input wires $a=1$, $c=0$ and output wires $i=j=1$).
The monomers corresponding to the correct computation on the given inputs
bind in the polymer with the input monomer.
Additional ``cap'' monomers (for example, $\{1_a^*, 0_c^*\}$)
and ``wrap-around'' site types ($w$ and $w^*$)
non-trivially ensure that correct computation is enforced in the stable configuration.
To verify correctness,
we would like to algorithmically confirm
that there is no stable configuration such that the wrong output monomer
is bound in the same polymer as the input.
Note that this TBN has about $10 ^ {61}$ saturated configurations.
Even checking a billion configurations per second,
checking all of them would take $10 ^ {44}$ years.

{\begin{figure}[t]
  \centering
  \includegraphics[width=1.0\textwidth]{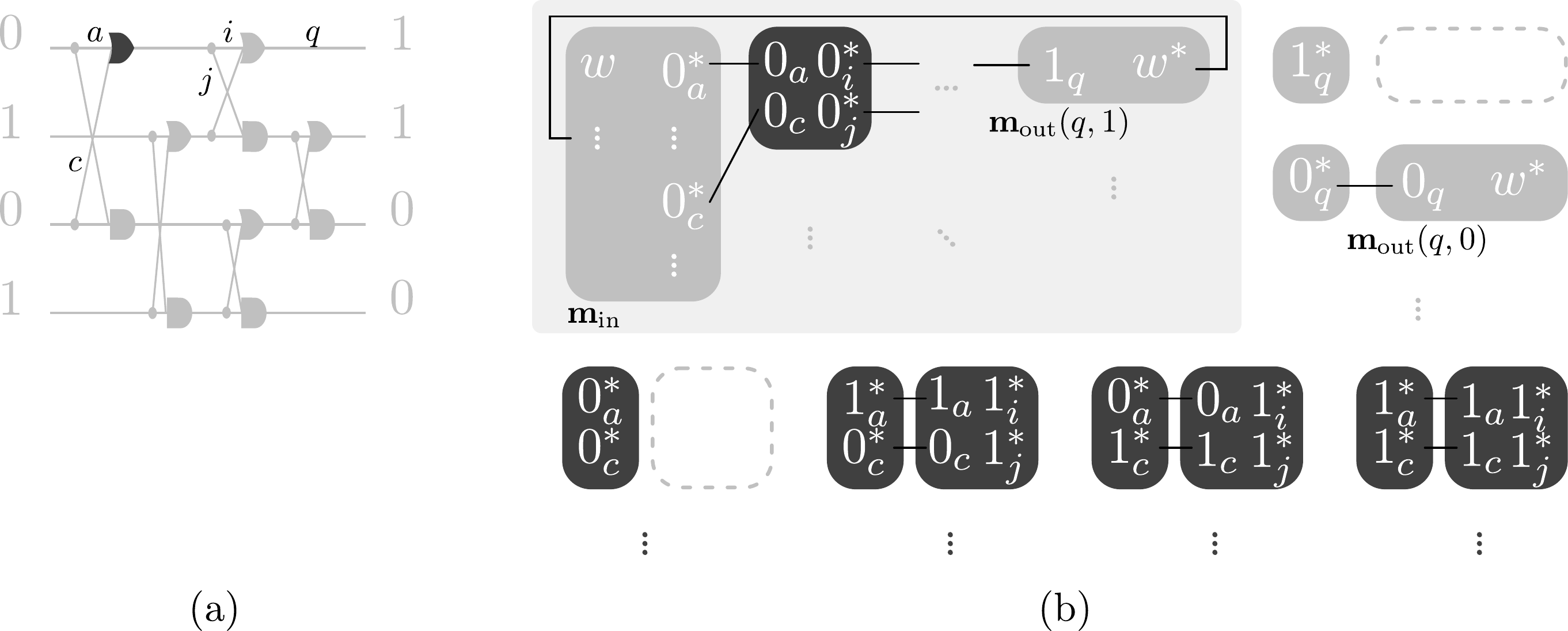}
  \caption{
    (a) A Boolean sorting network with four inputs
    and (b) a stable configuration of a TBN that simulates it.
    The full construction for arbitrary circuits is given in prior work~%
    \cite{chalk2018thermodynamically}.
    Each logic gate in (a) corresponds to a set of monomers in (b).
    For example,
    the black OR gate in (a)
    corresponds to the black monomers in (b).
    The $\monomer1_{\text{in}}$ monomer encodes the inputs to the circuit.
    The $\monomer1_{\text{out}}$ monomers correspond
    to possible outputs of the circuit.
    In a particular configuration,
    the computed output corresponds to which of the $\monomer1_{\text{out}}$
    monomers are in the same polymer as $\monomer1_{\text{in}}$.
    In the example configuration in (b),
    $\monomer1_{\text{out}}(q,1)$
    is in the same polymer as the input monomer,
    so the output on wire $q$ is $1$.
    The construction ensures that every stable configuration
    (for example, the configuration shown in (b)) computes the correct output.
  }
  \label{fig:circuit}
\end{figure}
}

Beyond theoretical constructions,
DNA strand displacement cascades have been widely used
to realize complex molecular reaction pathways in the laboratory%
~\cite{zhang2011dynamic}.
In a strand displacement cascade,
computation is carried out by combinations of DNA strands.
When certain input strands are present,
they can bind to DNA complexes and displace strands of the complex.
These newly displaced strands act as input strands for other complexes.
The cascade of displacements that follows can free certain output strands.
Which are free encodes the output of the computation.

We can use a TBN to confirm that the release of the output strand
incurs a configurational entropy penalty unless the correct input is present.
Figure~\ref{fig:abcd} shows two stable configurations
of a TBN that models a strand displacement cascade computing an AND of two inputs:
the output strand can be free in a stable configuration
if and only if both inputs are present.
This construction was introduced in prior work,
but the TBN model gives us a way of precisely articulating
the combinatorial penalty to falsely releasing the output%
~\cite{thachuk2015leakless,tbn}.%
\footnote{Strand displacement systems typically include weakly-binding domains known as toeholds,
which are responsible for fast initiation of the strand displacement reaction. 
Toehold binding is weak by design to avoid sequestering input strands in incorrect complexes.
Because toeholds are not necessarily expected to be bound in thermodynamically favored configurations, they fall outside our TBN abstraction (toehold domains are not shown in Fig.~\ref{fig:abcd}).
Nonetheless, since toeholds can substantially change the thermodynamic equilibrium,
more comprehensive arguments should ensure that the configurational entropy penalty is not offset by binding of additional toeholds.
This analysis depends on the specific strand displacement system in question and is outside the scope of this work.
}

In the remainder of the paper,
we formally define and resolve questions about TBNs such as these.
In Sections \ref{satconfig} and~\ref{stablyfree}
we focus on analyzing their computational complexity.
In Sections \ref{solver} and~\ref{solverblackbox} we develop a method
to answer them via a reduction to SAT.
The tool we develop handles the examples in this section
and confirms the answers we have discussed here.

{\begin{figure}[t]
  \renewcommand{\U}[1]{\falign{\spanel{\ensuremath{#1\vphantom{bg}}}}}%
  \renewcommand{\X}[1]{\falign{\spanel{\ensuremath{\pflip{#1}\vphantom{bg}}}}}%

  \quad
  \subfloat[]{
    \def\svgwidth{0.35\textwidth}
    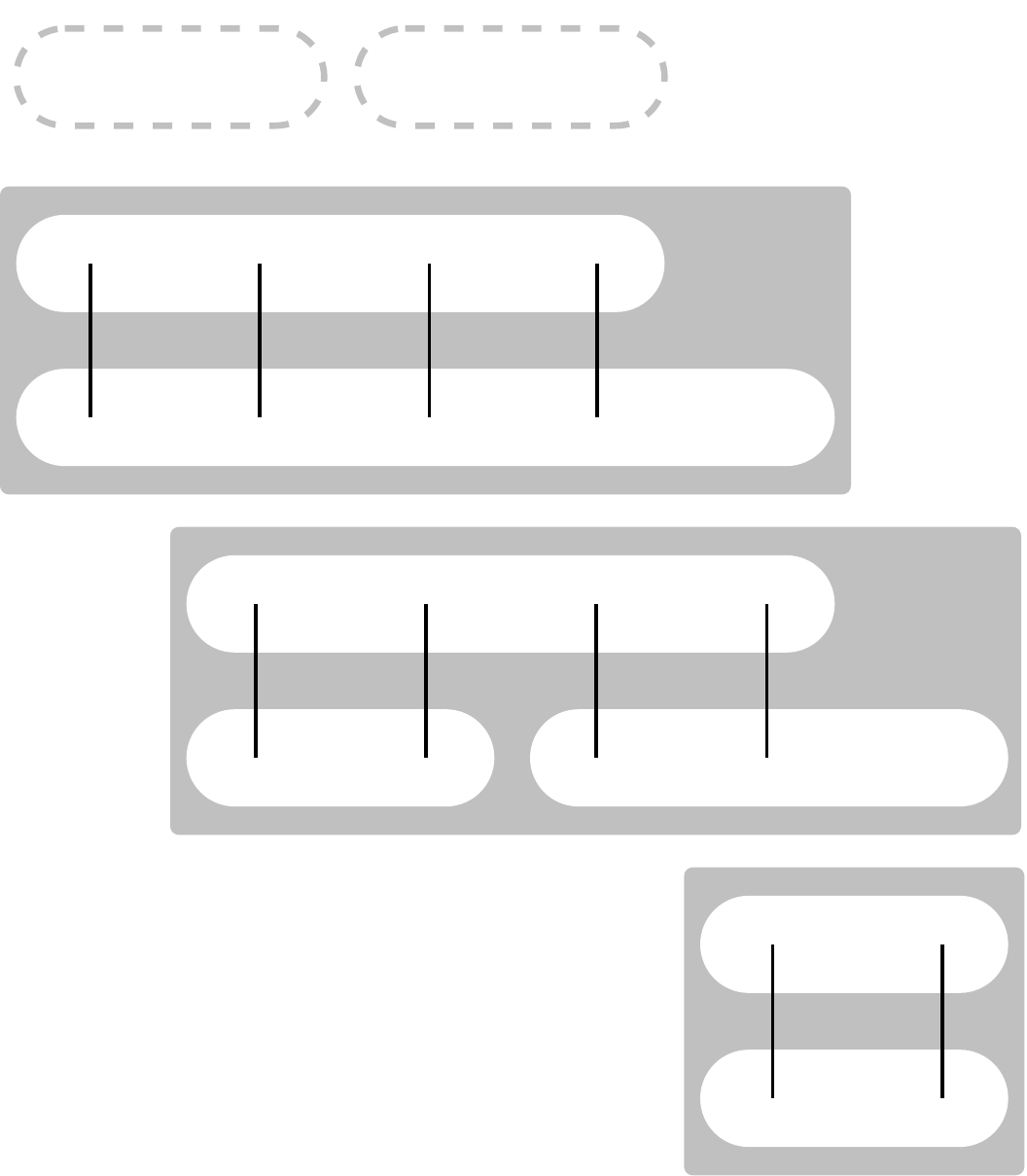
  }
  \hspace*{\fill}
  \subfloat[]{
    \def\svgwidth{0.35\textwidth}
    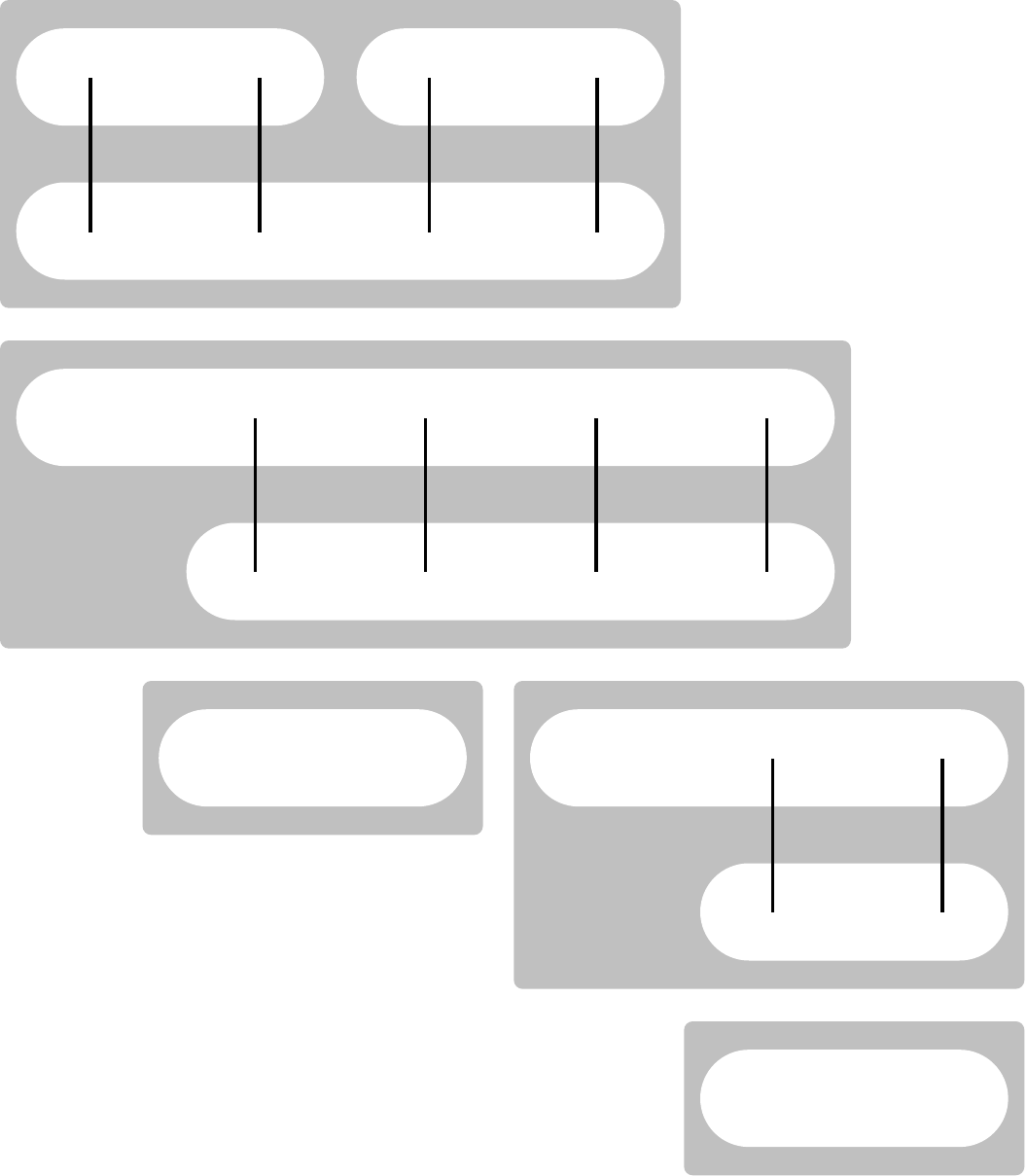
  }
  \quad
  \quad
  \quad

  \caption[]{A TBN for a strand displacement cascade module
  that behaves as an AND gate.
  (a) The gate without the input monomers.
  This is the unique stable configuration.
  The output monomer remains bound.
  (b) The same gate with the input monomers added.
  The configuration illustrated now becomes
  an additional stable configuration,
  and it has the output monomer free.
  A kinetic pathway that leads from the inputs being present
  to the outputs being free would correspond to
  flipping the bonds in a cascade from top to bottom.
  Note that when exactly one of the input monomers is present (not shown),
  the output monomer still is not free in any stable configuration.
  }
  \label{fig:abcd}
\end{figure}

}

\section{$\pcount$} \label{satconfig}

Stable configurations are of central interest,
so we need to be able to identify them.
More generally, we ask what is the maximum number of separate polymers achievable in a saturated configuration of a TBN, which we call $\stablepolycount(\therm1)$.
Knowing this quantity,
we can determine whether a given configuration is stable
by checking whether it has $\stablepolycount(\therm1)$ polymers.
We will see in later sections that other questions about TBNs
can be reduced to calculating $\stablepolycount(\therm1)$ as well.

\begin{definition}
  $(\therm1, k)$ is in $\pcount$ if
  some saturated configuration of the \tbn/ \hspace{0.2em}$\therm1$
  has at least $k$ polymers.
  $\stablepolycount(\therm1)$ is the greatest such $k$.
\end{definition}

\aftertheoremindent
In this section,
we prove that $\pcount$ is $\ccomplete{\cnp}$
and even that $\stablepolycount$ is hard to approximate.
To formalize our analysis,
we establish the encoding of a TBN as follows.
First, the encoding of a monomer
is a sequence of its sites (arbitrarily ordered).
Then the encoding of a TBN
is a sequence of the encoding of its monomers (also arbitrarily ordered).

\begin{claim} \label{satconfignp}
  $\pcount$ is in $\cnp$.
\end{claim}

\begin{proof}
  Consider a TBN $\therm1$ and an integer $k$.
  Suppose $\therm1$ has a saturated configuration $\config1$
  with at least $k$ polymers.
  Then $\config1$ is a certificate of polynomial size.
  Check that $\config1$ is saturated
  by checking that every two unpaired sites
  are not complements.
  Check that $\config1$ has $k$ polymers
  by counting the connected components.
  \qedhere
\end{proof}

\newcommand{\pexactcover}{\prob{ExactCover}}%
\aftertheoremindent
To show hardness we reduce from the decision problem $\pexactcover$,
one of Karp's original $\ccomplete{\cnp}$ problems.

\newcommand{\universe}{U}%
\newcommand{\bin}{B}%

\begin{definition}
  A set
  $X$
  of sets
  is in $\pexactcover$ if
  some subset of $X$
  partitions the universe $\bigcup X =
  \{ y \in x : x \in X \}$.
  Such a subset is called an \emph{exact cover}.
\end{definition}

\newcommand{\ctbn}[2][]{{\therm1}_{#1}(#2)}%
\newcommand{\muni}{\monomer1 _ \universe}%
\newcommand{\mrest}{\monomer1 _ {\bin \setminus \universe}}%

\aftertheoremindent
The reduction relies on a transformation
from an instance $X$ of $\pexactcover$ to a TBN $\ctbn{X}$,
which Figure~\ref{fig:trans} illustrates.
Let $\universe = \bigcup X$ be the universe and
$\bin = \biguplus X$ be the multiset sum of the sets in $X$.
Then $\ctbn{X}$ has three kinds of monomers:
$\monomer1_x = x$ for each set $x$ in $X$,
$\muni = \{\flip{\site1} : \site1 \in \universe\}$, and
$\mrest = \{\flip{\site1} : \site1 \in \bin \setminus \universe\}$.%
The number of polymers in a stable configuration of $\ctbn{X}$
depends on whether $X$ is a yes or no instance.

\begin{claim} \label{thm:Sj}
  \begin{equation}
    \stablepolycount(\ctbn{X}) =
    \begin{cases}
      2 &\text{if $X$ is in $\pexactcover$}\\
      1 &\text{otherwise}\\
    \end{cases}
  \end{equation}
\end{claim}

\begin{proof}
  The $\monomer1_x$ together
  perfectly complement
  $\muni$ and $\mrest$ together.
  So in a saturated configuration $\config1$,
  each $\monomer1_x$
  is completely bound
  to one or both of
  $\muni$ and $\mrest$.
  So $\config1$ has one or two polymers.

  Suppose $X$ has an exact cover $C$.
  Then $C$ partitions $\universe$.
  So bind $\monomer1_x$ to $\muni$
  for all $x$ in $C$.
  And bind $\monomer1_x$ to $\mrest$
  for all $x$ not in $C$.
  This produces a saturated configuration with two polymers.

  Conversely,
  suppose $\config1$ is a saturated configuration with two polymers.
  Then each $\monomer1_x$ is completely bound
  to $\muni$
  or completely bound to $\mrest$.
  Those $\monomer1_x$ bound to $\muni$
  then partition $\universe$
  and constitute an exact cover.
  \qedhere
\end{proof}

{\begin{figure}[t!]
  \newcommand{\da}[1]{{#1}_{1}}%
  \newcommand{\db}[1]{{#1}_{2}}%

  \hspace*{\fill}
  \subfloat[]{
    \centering
    \newcommand{\fsite}[1]{\ifcase#1\or a\or b\or c\else\@ctrerr\fi\vphantom{by}}%

    \def\svgwidth{0.25\columnwidth}
    \newcommand{\scm}[1]{\makebox[0pt][c]{$\monomer1_{#1}$}}%
    \newcommand{\slm}[1]{\makebox[0pt][l]{$\monomer1_{#1}$}}%
    \newcommand{\slmu}{\makebox[0pt][l]{$\muni$}}%
    \newcommand{\slmb}{\makebox[0pt][l]{$\mrest$}}%
    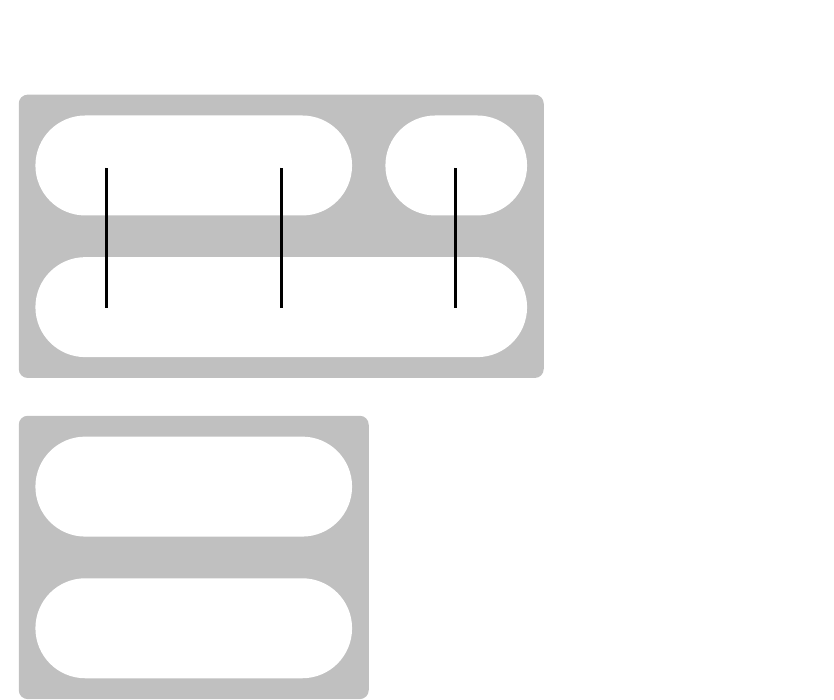
  }
  \hspace*{\fill}
  \subfloat[]{
    \renewcommand{\U}[1]{\falign{\spanel{\ensuremath{#1\vphantom{bg}}}}}%
    \renewcommand{\X}[1]{\falign{\spanel{\ensuremath{\pflip{#1}\vphantom{bg}}}}}%

    \def\svgwidth{0.35\columnwidth}
    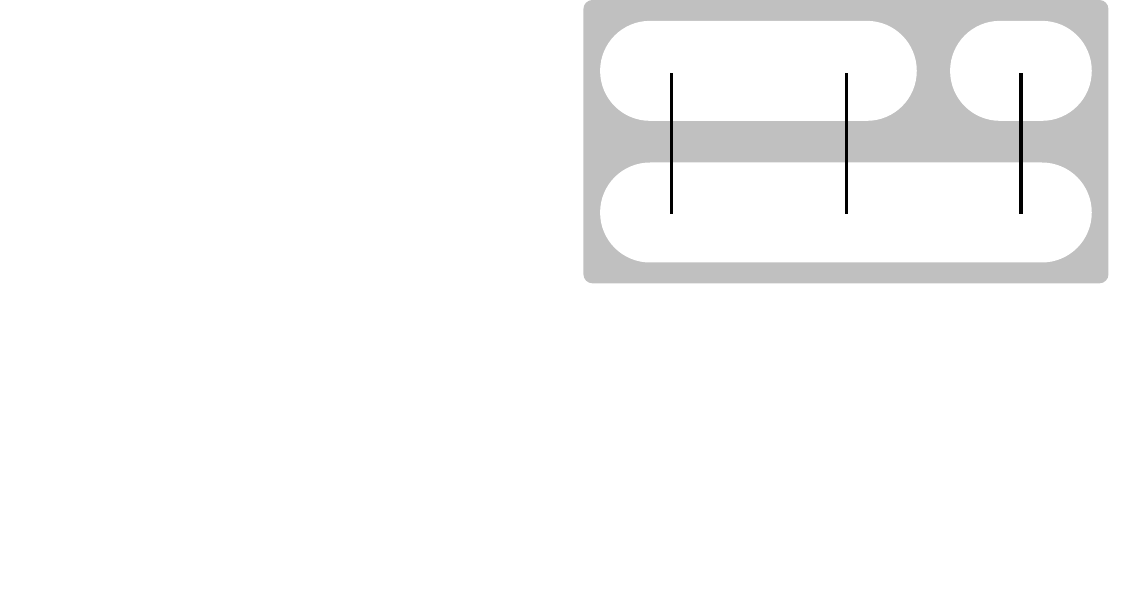
  }
  \hspace*{\fill}
  \caption{
    The transformation of the instance
    $X = \{\{a, b\}, \{b, c\}, \{c\}\}$ of $\pexactcover$
    into the TBNs $\ctbn{X}$ and $\ctbn[3]{X}$.
    (a) A saturated configuration of $\ctbn{X}$ with two polymers,
    showing that $X$ has an exact cover $\{\{a, b\}, \{c\}\}$.
    (b) A saturated configuration of $\ctbn[3]{X}$ with three polymers
    showing the same exact cover.
    }
  \label{fig:trans}
\end{figure}

}

\aftertheoremindent
This is enough to prove that $\pcount$
is $\ccomplete{\cnp}$,
but to prove a strong hardness of approximation result for $\stablepolycount$,
we scale the gap between the two cases
up to a factor of $j$.
To do so,
we transform $X$ into $\ctbn[j]{X}$
as Figure~\ref{fig:trans} illustrates.
We combine $j - 1$ indexed copies of $\ctbn{X}$ into a single TBN
and then merge
the $j - 1$ indexed copies of $\mrest$ into a single monomer.
This way $\ctbn{X}$ is the same as $\ctbn[2]{X}$,
and in general we have the following.

\begin{claim}
  \label{thm:STj}
  \begin{equation}
    \stablepolycount(\ctbn[j]{X}) =
    \begin{cases}
      j &\text{if $X$ is in $\pexactcover$}\\
      1 &\text{otherwise}\\
    \end{cases}
  \end{equation}
\end{claim}

\aftertheoremindent
Bonds form only within, not between, copies in $\ctbn[j]{X}$.
So Claim~\ref{thm:STj} holds by reasoning similar to
the proof of Claim~\ref{thm:Sj}.

We now see that $\stablepolycount$ is hard to approximate.
We define a one-sided approximation
in the way typical of maximization problems.

\begin{definition}
  $A$ is a $\rho$ factor approximation for $f$
  if $f(x) / \rho \leq A(x) \leq f(x)$
  for all $x$.
\end{definition}

\aftertheoremindent
A trivial $n$ factor approximation algorithm for $S$
is to always output~1.
This trivial algorithm turns out to be optimal
among polynomial factor approximations.

\begin{claim}
  No $n^{\delta}$ factor approximation algorithm
  for $\stablepolycount$
  runs in time polynomial in $n$
  for any $\delta < 1$
  unless $\cp = \cnp$,
  where $n$ is the number of monomers.
\end{claim}

\begin{proof}
  \newcommand{\apxalg}{A}%
  To the contrary,
  suppose $\apxalg$ is such an algorithm,
  and consider a set $X$ of $m$ sets.
  Then $\ctbn[j]{X}$ has
  $n = (m + 1)(j - 1) + 1 < 2 m j$ monomers.
  So choose $j > (2 m) ^ \frac{\delta}{1 - \delta}$.
  Raise both sides to $1 - \delta$
  and rearrange
  to see that $j > (2 m j) ^ {\delta} > n ^ \delta$.
  If $X$ is not in $\pexactcover$,
  then $\stablepolycount(\ctbn[j]{X}) = 1$.
  So
  \begin{align}
    \apxalg(\ctbn[j]{X})
      \leq {}& \stablepolycount(\ctbn[j]{X})
      = 1
    .
    \intertext{If $X$ is in $\pexactcover$,
    then $\stablepolycount(\ctbn[j]{X}) = j$.
    So
    }
    \apxalg(\ctbn[j]{X})
    \geq {}& S(\ctbn[j]{X}) / n ^ \delta
    {} = j / n ^ \delta
    \makebox[0pt][l]{${} > 1$.}
  \end{align}
  So $X$ is in $\pexactcover$ if and only if
  $\apxalg(\ctbn[j]{X}) > 1$.
  \qedhere
\end{proof}

\newcommand{\pPair}[2]{\operatorname{Pair}(#1, #2)}%
\newcommand{\pColor}[2]{\operatorname{Color}(#1, #2)}%
\newcommand{\pSum}[2]{\operatorname{Sum}(#1, #2)}%
\newcommand{\pBind}[2]{\operatorname{Bind}(#1, #2)}%
\newcommand{\pRep}[1]{\operatorname{Rep}(#1)}%

\section{Computing $\pcount$} \label{solver}

A hallmark property of an $\ccomplete{\cnp}$ decision problem is that
some instances will be hard to solve unless $\cp = \cnp$.
However,
there are still interesting instances that can be solved efficiently.
Since many real-world problems are $\ccomplete{\cnp}$,
various approaches have been developed
to perform well on interesting instances.
In this section,
we apply such an approach.
We encode the $\pcount$ problem
into $\psat$, the Boolean satisfiability problem,
which allows a $\psat$ solver to find a solution
that we can decode to obtain a saturated configuration.
By querying $\pcount$ for different $k$
we can compute $\stablepolycount(\therm1)$
as well as obtain a stable configuration.

\subsection{The Boolean satisfiability Problem}%

\begin{definition}
  A formula $\phi$ in propositional logic is called \emph{satisfiable}
  iff there exists an assignment to the Boolean variables in $\phi$ such
  that the formula evaluates to true. The Boolean satisfiability problem
  asks whether a given formula $\phi$ is satisfiable.
\end{definition}

\aftertheoremindent
A $\psat$ solver is a tool
that determines whether a formula has a satisfying assignment.
In the last two decades,
$\psat$ solvers have become powerful enough
to efficiently solve interesting instances of hard problems.
The approach has been successful in areas such
as hardware and software verification~\cite{BMC,Copty2001,SBMC}.

To use this approach to solve $\pcount$,
we will translate a \tbn/ $\therm1$ and an integer $k$
into a CNF formula $\phi$
such that satisfying assignments of $\phi$
correspond to
saturated configurations of $\therm1$ with at least $k$ polymers.
Recall that a CNF (conjunctive normal form) formula
is a conjunction of disjunctions,
such as $(x \lor \neg y) \land (x \lor z) \land z$.

\newcommand{\atmost}[2]{{\leq}_{#1} #2}%
\newcommand{\atleast}[2]{{\geq}_{#1} #2}%
\newcommand{\atmostimpl}[2]{\operatorname{LT}_{#1}(#2)}%
\newcommand{\atleastimpl}[2]{\operatorname{GT}_{#1}(#2)}%

\subsection{Encoding saturated configurations}

We construct a formula
where satisfying assignments correspond to
saturated configurations of a \tbn/ $\therm1$.
The formula uses a Boolean variable $\pPair{\site1}{\site2}$ for each pair of
complementary sites $\site1$ and $\site2$.
Assigning $\pPair{\site1}{\site2}$ true
will mean that $\site1$ and $\site2$ are paired.
Otherwise they are unpaired.
Note that pairing is symmetric,
so $\pPair{\site1}{\site2}$ and $\pPair{\site2}{\site1}$
are the same variable.

\newcommand{\compsites}[1]{\ensuremath{C(#1)}}%
In order to encode a valid configuration,
we add the constraint
\begin{equation}
  \atmost{1}{ \{\pPair{\site1}{\site2} : \site2 \in \compsites{\site1}\} }
\end{equation}
to the formula for each site $\site1$,
where $\compsites{\site1}$ is the sites in $\therm1$
complementary to $\site1$.
The direct encoding of $\atmost{1}$ (read ``at most one'')
includes a binary clause
$\lnot \pPair{\site1}{\site2} \lor  \lnot \pPair{\site1}{\site3}$
for each two sites $\site2$ and $\site3$ in $\compsites{\site1}$.
The number of such clauses
is quadratic in the size of $\compsites{\site1}$.
Although there do exist encodings that consist
of only a linear number of binary clauses~\cite{Sinz},
due to the limited size of $C(s)$ in our test suite
we found that the direct encoding works best.

\begin{claim}
  A configuration is saturated iff every limiting site is paired.
\end{claim}

\newcommand{\limsites}{L}%
\aftertheoremindent
We omit the proof,
but this fact allows us to easily encode saturation.
To do so, we add the constraint
\begin{equation}
  \bigwedge_{s \in \limsites}
  \atleast{1}{\{\pPair{\site1}{\site2} : \site2 \in \compsites{\site1}\}},
\end{equation}
to the formula where $L$ is the set of limiting sites of $\therm1$.
Notice that each $\atleast{1}$ (read ``at least one'') constraint
is simply a clause.

\subsection{Encoding polymers}

To begin identifying polymers,
we convert site pairing to monomer binding using Boolean variables
like $\pBind{\monomer1}{\monomer2}$.
For now, assigning $\pBind{\monomer1}{\monomer2}$ true
will mean that monomers $\monomer1$ and $\monomer2$ are bound.
Note that binding is also symmetric,
so
$\pBind{\monomer1}{\monomer2}$ and $\pBind{\monomer2}{\monomer1}$ 
are the same variable.
To convert pairing to binding,
we add the constraint
\begin{equation}
  \pPair{\site1}{\site2} \to \pBind{\monomer1}{\monomer2}
\end{equation}
to the formula for each site $\site1$ in a monomer $\monomer1$
and each complementary site $\site2$ in a different monomer $\monomer2$.

We expand $\pBind{\cdot}{\cdot}$ to be transitive
by adding the constraint
\begin{equation}
  \pBind{\monomer1}{\monomer2} \land
  \pBind{\monomer1}{\monomer3} \to
  \pBind{\monomer2}{\monomer3}
\end{equation}
to the formula
for every three distinct monomers
$\monomer1$,
$\monomer2$,
$\monomer3$.
The addition of transitivity ensures that monomers
$\monomer1$ and $\monomer2$ are part of the same polymer iff
$\pBind{\monomer1}{\monomer2}$ is true.

\subsection{Maximizing the number of polymers}

The most involved part of the encoding is enforcing
the configuration to have at least $k$ polymers.
There are various ways to encode such a cardinality constraint,
and the quality of that encoding can make or break
our entire approach.
Several techniques have been developed to automatically
improve a low quality encoding~\cite{BVE,BVA}.
However, it typically pays to manually optimize it.
We implemented several encodings,
and describe here the one that resulted in the best performance,
which was inspired by the representative encoding~\cite{HeuleS15}.

Let $\monomer1_1, \monomer1_2, \ldots, \monomer1_n$
be an arbitrary ordering of the $n$ monomers of $\therm1$.
We call a monomer $\monomer1$ the \emph{representative}
of a polymer $\polymer1$ iff
every other monomer in $\polymer1$ follows $\monomer1$ in the ordering.
To determine the representatives,
we add the constraint
\begin{equation}
  \pBind{\monomer1}{\monomer2} \to \lnot \pRep{\monomer2}
\end{equation}
to the formula for each monomer $\monomer1$
that precedes each monomer $\monomer2$.

Now, we can use these Boolean variables
to enforce that a configuration
has at least $k$ polymers
by adding
\begin{equation}
  \atleast{k}{\{\pRep{\monomer1} : \monomer1 \in \therm1\}}
\end{equation}
to the formula.
To effectively encode this cardinality constraint,
we introduce new Boolean variables $\pSum{i}{j}$
for $1 \leq i \leq n$ and $1 \leq j \leq k$,
which will mean that among $\monomer1_1, \ldots, \monomer1_i$
there are at least $j$ representatives.
The encoding,
which is a simplification of~\cite{Sinz},
is shown below and also illustrated in Figure~\ref{fig:card}.
\begin{align}
  \lnot \pSum{i}{j} &\to  \lnot \pSum{i+1}{j+1} \\
  \lnot \pSum{i}{j} \land \pSum{i+1}{j} &\to \phantom{\lnot}\pRep{\monomer1_{i + 1}} \\
  \label{repfirst}
  \lnot \pSum{1}{2} \\
  \label{repenough}
  \phantom{\lnot}\pSum{n}{k}&.
\end{align}

\begin{figure}
\centering
\begin{tikzpicture}
\draw[step=0.5cm,black,very thin] (0,0) grid (8,3);
\draw[step=0.5cm,black,very thin] (0,-1) grid (8,-.5);

\node[draw=none, left] at (0,-0.75) {$\pRep{1}$};
\node[draw=none, left] at (0,0.75) {$\pSum{1}{2}$};
\node[draw=none, right] at (8,-0.75) {$\pRep{n}$};
\node[draw=none, right] at (8,2.75) {$\pSum{n}{k}$};

\node[draw=none] at (0.25,0.75) {{\bf 0}};
\node[draw=none] at (0.25,1.25) {$\times$};
\node[draw=none] at (0.25,1.75) {$\times$};
\node[draw=none] at (0.25,2.25) {$\times$};
\node[draw=none] at (0.25,2.75) {$\times$};
\node[draw=none] at (0.75,1.25) {{0}};
\node[draw=none] at (0.75,1.75) {$\times$};
\node[draw=none] at (0.75,2.25) {$\times$};
\node[draw=none] at (0.75,2.75) {$\times$};
\node[draw=none] at (1.25,1.75) {{0}};
\node[draw=none] at (1.25,2.25) {$\times$};
\node[draw=none] at (1.25,2.75) {$\times$};
\node[draw=none] at (1.75,2.25) {{0}};
\node[draw=none] at (1.75,2.75) {$\times$};
\node[draw=none] at (2.25,2.75) {{0}};
\node[draw=none] at (7.75,2.75) {{\bf 1}};
\node[draw=none] at (7.75,2.25) {$\times$};
\node[draw=none] at (7.75,1.75) {$\times$};
\node[draw=none] at (7.75,1.25) {$\times$};
\node[draw=none] at (7.75,0.75) {$\times$};
\node[draw=none] at (7.75,0.25) {$\times$};
\node[draw=none] at (7.25,2.25) {{1}};
\node[draw=none] at (7.25,1.75) {$\times$};
\node[draw=none] at (7.25,1.25) {$\times$};
\node[draw=none] at (7.25,0.75) {$\times$};
\node[draw=none] at (7.25,0.25) {$\times$};
\node[draw=none] at (6.75,1.75) {{1}};
\node[draw=none] at (6.75,1.25) {$\times$};
\node[draw=none] at (6.75,0.75) {$\times$};
\node[draw=none] at (6.75,0.25) {$\times$};
\node[draw=none] at (6.25,1.25) {{1}};
\node[draw=none] at (6.25,0.75) {$\times$};
\node[draw=none] at (6.25,0.25) {$\times$};
\node[draw=none] at (5.75,0.75) {{1}};
\node[draw=none] at (5.75,0.25) {$\times$};
\node[draw=none] at (5.25,0.25) {{1}};
\node[draw=none] at (3.25,1.25) {{0}};
\node[draw=none] at (3.75,1.25) {{1}};
\node[draw=none] at (3.75,-0.75) {{1}};

\end{tikzpicture}
\caption{
  An illustration of the $\pSum{\cdot}{\cdot}$ variables
  (above) and the $\pRep{\cdot}$ variables (below)
  for a $\pcount$ problem with $n = 16$ monomers and $k = 6$ polymers.
  Variables that are assigned to false/true are shown as 0/1.
  The zeros propagate up diagonally, while the ones propagate down diagonally.
  A 0,1 pair in a row (above) implies that the monomer at the position of the 1
  is a representative (below).
  The bold entries indicate the unit clauses,
  while $\times$ indicates out redundant variables.
}
\label{fig:card}
\end{figure}

\newcommand{\nsites}{n}%
Overall,
the size of the formula we generate is bounded
in terms of the number $\nsites$ of sites in a TBN.
$O(\nsites^2)$ clauses of size $O(\nsites)$
encode a saturated configuration.
$O(\nsites^3)$ clauses of size $O(1)$
encode polymers, and
$O(\nsites k)$ clauses of size $O(1)$
check for $k$ polymers.
This gives us a total of $O(\nsites^3)$ clauses
of size $O(\nsites)$
(since $k \leq n$).

\subsection{Example of fast solving}
\label{sec:tree}

In principle,
to decide whether a \tbn/ has a saturated configuration
with at least a certain number of polymers,
we can simply check each saturated configuration in turn.
In practice,
the number of saturated configurations is typically astronomically large%
---for example, recall that the TBN in Fig.~\ref{fig:circuit}
has about $10 ^ {61}$ saturated configurations.
Thus any naive approach that somehow enumerates the saturated configurations
to identify the ones with the greatest number of polymers is bound to fail.
Yet in under a minute,%
\footnote{%
  For exact timing on this and other examples,
  see our software package and the documentation of the examples in it%
  ~\cite{solver}.}
our $\psat$ solving approach finds $S(\therm1)$
for the example and inputs in Fig.~\ref{fig:circuit},
and also checks $\ptogether$,
a harder problem that we define in Section~\ref{stablyfree}.
It does this using the Glucose $\psat$ solver
built and run with default options%
~\cite{audemard2009predicting}.%

\newcommand{\graph}{G}%
\newcommand{\target}{m}%
\newcommand{\graphtbn}{\therm1_\graph}%
\newcommand{\graphm}{\monomer1_\graph}%
\newcommand{\vertexm}[1]{\monomer2_#1}%
\newcommand{\targetm}{\vertexm{\target}}%

\section{Beyond $\pcount$} \label{stablyfree}

We will want to compute
more complicated properties of stable configurations
than just the number of polymers.

\begin{definition}
  \newcommand{\resp}{}%
  In a TBN $\therm1$,
  a property \defterm{can} (\resp{}\defterm{can stably}) hold
  if it holds for some saturated (\resp{}stable) configuration.
  $(\therm1, x)$ is in $\stablyq$
  if the property $Q_x$ can stably hold in $\therm1$.
\end{definition}

\aftertheoremindent
Even when $Q_x$ is an easy property to check,
$\stablyq$ seems harder than $\cnp$.
For $\pcount$, a saturated configuration with $k$ polymers
served as a certificate.
For $\stablyq$, a stable configuration satisfying $Q_x$
might seem to serve as a certificate.
But checking that a configuration is saturated is easy,
while checking that it is stable seems hard.
So we look to a class larger than $\cnp$.

$\cpnp$ is the class of problems decided
by a deterministic polynomial-time Turing machine
with access to an oracle in $\cnp$.
The Turing machine can alternate
between computing and querying its oracle adaptively.
If we instead require the machine to make all of its queries
in parallel as a single group,
then we get the class $\cparallelnp$.%
\footnote{$\cparallelnp$ is also written $\cp^{\cnp[\log]}$}
This class is large enough to contain $\stablyq$ for many $Q$.

\begin{claim} \label{stablyqinparallelnp}
  If checking $Q_x$ on configuration $\config1$ is in $\cnp$,
  then $\stablyq$ is in $\cparallelnp$.
\end{claim}

\begin{proof}
  \newcommand{\resp}{}%

  \newcommand{\satwith}[1]{A ^ {#1}}%
  \newcommand{\propwith}[1]{B ^ {#1}}%

  Consider a TBN $\therm1$ and parameter $x$.
  Let $\satwith{k}$ (\resp{}and let $\propwith{k}$)
  mean that some saturated configuration of $\therm1$ has $k$ polymers
  (\resp{}and has $Q_x$).
  We know $\satwith{k}$ is in $\cnp$ by Claim~\ref{satconfignp}.
  And we see $\propwith{k}$ is in $\cnp$
  by the witness $(\config1, w)$,
  where $w$ is in turn a witness that $Q_x$ holds for $\config1$.
  So query $\satwith{k}$ and $\propwith{k}$ for each $k$
  from 1 to the number of monomers in $\therm1$.
  Then the largest $k$ where $\satwith{k}$ holds
  is $\stablepolycount(\therm1)$.
  And $\propwith{\stablepolycount(\therm1)}$
  holds iff $(\therm1, x)$ is in $\stablyq$.
  \qedhere
\end{proof}

\aftertheoremindent
In the rest of this section we define two important properties $Q_x$
and show that $\stablyq$ is a complete problem for $\cparallelnp$ in both cases.
This uses mathematical tools involving graphs,
which we introduce first.

\subsection{Graph problems for reductions}

\begin{figure}[t]
  \centering
  \includegraphics[width=0.76\textwidth]{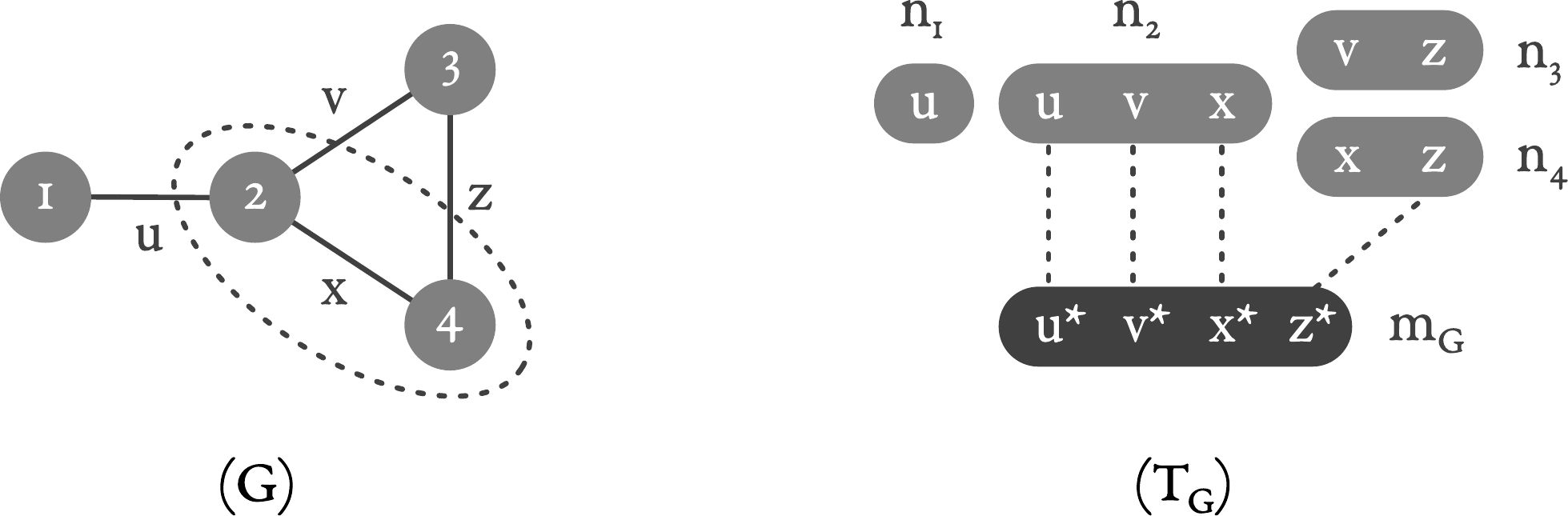}
  \caption{
    A graph $\graph$ and its transformation into a TBN $\graphtbn$.
    $\graphtbn$ has a monomer $\graphm$
    consisting of the edges of $\graph$,
    each complemented.
    For each vertex $v$ of $\graph$,
    $\graphtbn$ has a monomer
    $\vertexm{v}$
    consisting of the edges incident to $v$.
    Inside the dashed ellipse is a vertex cover of $\graph$.
    Outside the dashed ellipse is an independent set of $\graph$.
    The dashed lines indicate a corresponding
    saturated configuration of $\graphtbn$.
  }
  \label{fig:graphtbn}
\end{figure}

{
\newcommand{\comp}{\prob{Prob}}%

To see that $\stablyq$ is complete for $\cparallelnp$,
we show that there is a polynomial-time many-one reduction
from some complete problem $\comp$ to $\stablyq$,
denoted $\comp \reducesto \stablyq$.
Here $\comp$ will be a graph problem.
}

{
\newcommand{\sub}{X}%

Consider a subset $\sub$ of the vertices of a graph $\graph$.
Recall that
$\sub$ is an independent set if $\sub$ contains no neighbors.
$\sub$ is a vertex cover if $\sub$ contains an endpoint of each edge.
With this we define four key decision problems.
}

\begin{definition}
  \newcommand{\resp}{}%
  $(\graph, \target)$ is in $\pindrmember$ if
  some maximum independent set
  of the graph $\graph$
  (\resp{}does not) contain
  the vertex $\target$.
  $\pcoverrmember$ is analogous.
\end{definition}

\aftertheoremindent
Prior work shows that $\pcovermember$
is complete for $\cparallelnp$~\cite{Spa2005}.
We will later see that $\pindmember$
is as well.
And by a well-known duality between independent sets and vertex covers,
$\pindmember = \pcovernotmember$
and $\pindnotmember = \pcovermember$.
So these problems are all complete for $\cparallelnp$.

These graph problems connect to stable configurations of a TBN
by the construction in Fig.~\ref{fig:graphtbn}
and the following fact about it.

\begin{claim} \label{bindmember}
  $(\graph, \target)$ is in $\pindrmember$ iff
  $\graphm$ is not (is) bound to $\targetm$
  in some stable configuration of $\therm1_\graph$.
\end{claim}

\begin{proof}%
  \newcommand{\ind}{I}%
  \newcommand{\conf}{\config1}%
  \newcommand{\vertex}[1]{\ifcase#1\or v\or u\else\@ctrerr\fi}%
  Suppose some independent set $\ind$ of $\graph$ of size $k$
  does not (does) contain $\target$.
  Then binding $\graphm$ to $\vertexm{\vertex1}$
  for each $\vertex1$ not in $\ind$
  yields a saturated configuration of $\graphtbn$ with $k + 1$ polymers
  and $\graphm$ not bound (bound) to $\targetm$.

  Conversely,
  suppose some saturated configuration $\conf$ of $\graphtbn$
  has $k + 1$ polymers
  and $\graphm$ not bound (bound) to $\targetm$.
  Then collecting each vertex $\vertex1$
  where $\vertexm{\vertex1}$ is not bound to $\graphm$ in $\conf$
  yields an independent set of $\graph$ of size $k$
  not containing (containing) $\target$.
  \qed
\end{proof}

\subsection{$\ptogether$}

In a \tbn/ that implements a computation,
whether two specific monomers are together in the same polymer 
can be interpreted as a bit of output.
For example, this is the notion of output in the circuit construction shown in Fig.~\ref{fig:circuit}~
\cite{chalk2018thermodynamically}.
Note that this representation is experimentally meaningful when
usual fluorescence readout methods are used to interrogate system state:
A ``fluorophore'' moiety on one monomer emits signal measured by a spectrophotometer if not quenched by a ``quencher'' moiety on another monomer.
Thus the decrease in measured fluorescence corresponds to the fluorophore and quencher-labelled monomers being together in the same monomer. 

\begin{definition}
  Two monomers are \defterm{together} in a configuration
  if they are in the same polymer.
  $(\therm1, \monomer1_1, \monomer1_2)$ is in $\ptogether$ if
  $\monomer1_1$ and $\monomer1_2$ can be stably together in $\therm1$.
\end{definition}

\begin{claim} \label{togetherinparallelnp}
  $\ptogether$ is in $\cparallelnp$.
\end{claim}

\begin{proof}
  Checking that monomers are together is in $\cnp$,
  so Claim~\ref{stablyqinparallelnp} applies.
  \qedhere
\end{proof}

\begin{claim}
  $\pindnotmember \reducesto \ptogether$.
\end{claim}

\begin{proof}
  Consider the following statements.
  \begin{enumerate}
    \item $(\graph, \target)$ is in $\pindnotmember$.
    \item $\graphm$ is bound to $\targetm$
      in some stable configuration of $\therm1_\graph$.
    \item $(\therm1_\graph, \graphm, \targetm)$ is in $\ptogether$.
  \end{enumerate}
  (1) iff (2) by Claim~\ref{bindmember}.
  (2) iff (3) by definition.
  \qed
\end{proof}

\subsection{\pfree}

\aftertheoremindent
In typical strand displacement systems,
signals are carried by ``signal strands'' from one complex to another,
and signal strands are activated or inactivated
by being free or bound within a larger complex%
~\cite{zhang2011dynamic}.
Recall the example AND gate shown in Fig.~\ref{fig:abcd}:
the designated output monomer can be free in a stable configuration
iff both designated input monomers are present.
Predicting such behavior is important
to verifying the thermodynamic properties of strand displacement systems%
~\cite{thachuk2015leakless}.
Notice that $\pfree$ is not some kind of complement of $\ptogether$.
For example, in some cases a monomer can be stably free
and yet also can be stably together with other monomers.

\begin{definition}
  A monomer is \defterm{free} in a configuration
  if no other monomer binds to it.
  $(\therm1, \monomer1)$ is in $\pfree$ if
  $\monomer1$ can be stably free in $\therm1$.
\end{definition}

\begin{claim}
  $\pfree$ is in $\cparallelnp$.
  \label{freeinparallelnp}
\end{claim}

\begin{proof}
  Checking that a monomer is free is in $\cnp$,
  so Claim~\ref{stablyqinparallelnp} applies.
  \qedhere
\end{proof}

\begin{claim}
  $\pindmember \reducesto \pfree$.
\end{claim}

\begin{proof}
  Consider the following statements.
  \begin{enumerate}
    \item $(\graph, \target)$ is in $\pindmember$.
    \item $\targetm$ is not bound to $\graphm$
      in some stable configuration of $\therm1_\graph$.
    \item $\targetm$ is free in some stable configuration of $\therm1_\graph$.
    \item $(\therm1_\graph, \targetm)$ is in $\pfree$.
  \end{enumerate}
  (1) iff~(2) by Claim~\ref{bindmember}.
  (2) iff~(3)
  since $\targetm$ is compatible only with $\graphm$.
  (3) iff~(4) by definition.
  \qed
\end{proof}

\aftertheoremindent
Our last step is now to show that $\pindmember$ is complete.
So far we only know that $\pindnotmember$ is.%
\footnote{
  If these two problems were complements,
  then we would be done:
  \newcommand{\classx}{\class{C}}%
  since $\cparallelnp = \mathrm{co}\cparallelnp$,
  $A$ is complete for $\cparallelnp$ iff
  $\overline{A}$ is complete for $\mathrm{co}\cparallelnp$.
  However, the two problems are not complements. 
  For example,
  the graph consisting of a single edge between vertices $u$ and $v$
  has $v$ in the maximum independent set $\{v\}$
  and \emph{not} in the maximum independent set $\{u\}$.
}

\begin{claim} \label{notmember}
  $\pindnotmember \reducesto \pindmember$.
\end{claim}

\begin{proof}
  Consider a graph $\graph$ and a target vertex $\target$.
  Form $\graph'$ by first splitting each vertex.
  This means replace each vertex $u$
  with $\dot u$ and $\ddot u$
  \newcommand{\edge}[2]{#1 {-} #2}%
  and replace each edge $\edge{u}{v}$ with
  $\edge{\dot  u}{\dot  v}$,
  $\edge{\ddot u}{\ddot v}$,
  $\edge{\dot  u}{\ddot v}$,
  $\edge{\ddot u}{\dot  v}$.
  This way $\dot u$ and $\ddot u$ have the same neighbors.
  Then connect a new vertex $\target'$ to $\dot \target$ and $\ddot \target$.
  We will show that $(\graph, \target)$ is in $\pindnotmember$
  iff $(\graph', \target')$ is in $\pindmember$.

  \newcommand{\maximal}[1]{[ #1 ]}%
  \newcommand{\maximaloriginal}{\maximal{\graph}}%
  \newcommand{\maximalsplit}{\maximal{\graph'}}%

  Let $[H]$ denote the maxi\emph{mal} independent sets of a graph $H$.
  We claim that the mapping
  $f : \maximaloriginal \to \maximalsplit$ is a bijection
  when defined by
  \begin{align}
    f(I)
    = \begin{cases}
      \dot I \cup \ddot I &\text{if } \target \in I \\
      \dot I \cup \ddot I \cup \{\target'\} &\text{else.}
    \end{cases}
  \end{align}
  To see that $f$ is onto,
  consider $I'$ in $\maximalsplit$.
  If $\dot u$ is in $I'$,
  then by independence,
  no neighbor of $\dot u$ is in $I'$.
  By construction, these are also the neighbors of $\ddot u$.
  So by maximality, $\ddot u$ is in $I'$.
  So $I' \setminus \{ \target' \} = \dot I \cup \ddot I$ for some $I$.
  This way $\target$ is in $I$
  iff $\dot \target$ and $\ddot \target$ are in $I'$.
  By independence and maximality,
  this is iff $\target'$ is not in $I'$.
  So $I' = f(I)$.
  If we could add some $u$ to $I$,
  then we could add $\dot u$ to $I'$,
  but $I'$ is maximal.
  So $I$ is in $\maximaloriginal$.

  Now notice that $\size{f(I)}$ is $2 \size{I}$
  or $2 \size{I} + 1$.
  So $\size{I} < \size{J}$
  if and only if $\size{f(I)} < \size{f(J)}$.
  So $I$ is a maxi\emph{mum} independent set
  if and only if $f(I)$ is.
  And $\target$ is not in $I$ if and only if
  $\target'$ is in $f(I)$.
  \qedhere
\end{proof}

\aftertheoremindent
Intuitively,
splitting the graph in two
at the start of the proof
makes it so that adding the one additional node
preserves the independent sets,
which is necessary for the inequality at the end to hold.

Besides $\ptogether$ and $\pfree$,
many other interesting properties of stable configurations
are probably also $\ccomplete{\cparallelnp}$ to decide.

\section{$\pcount$ as a black box} \label{solverblackbox}

\newcommand{\maxpolys}[2][]{\stablepolycount_{#1}(#2)}%
\newcommand{\polycount}[1]{\size{#1}}%

Since $\stablyq$ is in $\cparallelnp$ (for any $Q$ in $\cnp$), we know we can in principle use any NP-complete problem as a black box to decide $\stablyq$.
In this section, we show that for $\ptogether$ and $\pfree$ in particular, we can use $\pcount$ as a block box in a very simple way.
The reductions below solve $\ptogether$ and $\pfree$ with two calls to $\maxpolys{\therm1}$.
Since our accompanying software package already computes $\maxpolys{\therm1}$ as described in Section~\ref{solver}, we additionally implement these reductions to compute $\ptogether$ and $\pfree$~\cite{solver}.

\subsection{Computing $\pfree$}

\newcommand{\maxfreepolys}[1]{\maxpolys[\monomer1\,\mathrm{free}]{#1}}%

\begin{definition}
  $\maxfreepolys{\therm1}$
  is the most polymers in any saturated configuration of $\therm1$
  with monomer $\monomer1$ free.
\end{definition}

\aftertheoremindent
Now $(\therm1, \monomer1)$ being in $\pfree$
is equivalent to
$\maxfreepolys{\therm1} = \maxpolys{\therm1}$.
So we just need a way to compute $\maxfreepolys{\therm1}$.

\begin{claim} \label{freetest}
  If monomer $\monomer1$ can be free in $\therm1$,
  then
  $\maxfreepolys{\therm1}
  =
  \stablepolycount(\therm1 \setminus \{\monomer1\}) + 1$.
\end{claim}

\begin{proof}
  \newcommand{\removed}{\therm1'}%
  Suppose $\monomer1$ can be free in $\therm1$,
  and let $\removed = \therm1 \setminus \multiset{\monomer1}$.
  Among saturated configurations of $\therm1$ with $\monomer1$ free,
  let $\config1$ be one with the most polymers.
  Let $\config1'$ be $\config1$ with $\monomer1$ removed.
  Then $\config1'$
  is a saturated configuration of $\removed$.
  So
  \begin{align}
    \maxfreepolys{\therm1}
    = \size{\config1}
    = \size{\config1'} + 1
    \leq \maxpolys{\removed} + 1.
  \end{align}
  Conversely, let $\config1'$ be a stable configuration of $\removed$.
  Let $\config1$ be $\config1'$ with $\monomer1$ added
  as a polymer.
  $\monomer1$ can be free
  and $\config1'$ is already saturated,
  so no additional bonds can form in $\config1$.
  So $\config1$ is a saturated configuration of $\therm1$
  with $\monomer1$ free.
  So
  \begin{align}
    \maxfreepolys{\therm1}
    \geq \size{\config1}
    = \size{\config1'} + 1
    = \maxpolys{\removed} + 1.
  \end{align}
  So overall, $\maxfreepolys{\therm1} = \maxpolys{\removed} + 1$.
  \qedhere
\end{proof}

\aftertheoremindent
And $\monomer1$ can be free in $\therm1$
if none of its sites are limiting,
which is easy to check.
So we can easily compute $\pfree$ by comparing $\maxpolys{\therm1}$ and $\maxpolys{\therm1 \setminus \{\monomer1\}}$.

\subsection{Computing $\ptogether$}

We can now use $\pfree$ in turn as a black box to decide $\ptogether$.

{
\newcommand{\tbnbase}{\therm1}%
\newcommand{\ma}{\monomer1_1}%
\newcommand{\mb}{\monomer1_2}%
\newcommand{\mf}{\monomer2}%
\newcommand{\mfx}{\flip{\mf}}%

\begin{claim}
  The following are equivalent.
  \begin{enumerate}
    \item $(\tbnbase \uplus \{\ma, \mb\}, \ma, \mb)$
      is in $\ptogether$.
    \item $(\tbnbase \uplus \{\ma', \mb', \mfx, \mf\}, \mf)$
      is in $\pfree$.
  \end{enumerate}
  Here $\ma' = \ma \uplus \{x\}$,
  $\mb' = \mb \uplus \{y\}$,
  $\mfx = \{\flip{x}, \flip{y}\}$,
  $\mf = \{x, y\}$,
  $x$ and $y$ are new site types,
  and $\uplus$ indicates the multiset sum.
\end{claim}

\begin{proof}
  \newcommand{\map}[1]{#1 _ {+}}%
  \newcommand{\unmap}[1]{#1 _ {-}}%
  \newcommand{\tbnorig}{\unmap{\therm2}}%
  \newcommand{\tbnadd}{\map{\therm2}}%
  Let $\tbnorig = \tbnbase \uplus \{\ma, \mb\}$
  and $\tbnadd = \tbnbase \uplus \{\ma', \mb', \mfx, \mf\}$.

  From a saturated configuration $\config1$ of $\tbnorig$,
  form the saturated configuration $\map{\config1}$ of $\tbnadd$
  as follows.
  Replace $\ma$ with $\ma'$ and $\mb$ with $\mb'$.
  If $\ma$ and $\mb$ are together,
  bind $\ma'$ and $\mb'$ to $\mfx$ and add the polymer $\{\mf\}$.
  Otherwise add the polymer $\{\mfx, \mf\}$.
  This way
  \begin{equation}
    \size{\map{\config1}} = \size{\config1} + 1
    .
  \end{equation}

  From a saturated configuration $\config2$ of $\tbnadd$,
  form the saturated configuration $\unmap{\config2}$ of $\tbnorig$
  as follows.
  Replace $\ma'$ with $\ma$ and $\mb'$ with $\mb$.
  Remove $\mfx$ and $\mf$.
  Since $\config2$ is saturated,
  $\mfx$ is bound to some monomer,
  so
  \begin{equation}
    \size{\unmap{\config2}} \geq \size{\config2} - 1
    .
  \end{equation}

  To see that (1) implies (2),
  suppose $\ma$ and $\mb$ are together
  in a stable configuration $\config1$ of $\tbnorig$.
  We claim that $\map{\config1}$ is a stable configuration of $\tbnadd$ with $\mf$ free.
  To see that $\map{\config1}$ is stable,
  consider a saturated configuration $\config2$ of $\tbnadd$.
  Since $\config1$ is stable, $\size{\config1}  \geq \size{\unmap{\config2}}$, and we have
  \begin{equation}
    \size{\map{\config1}}
    = \size{\config1} + 1
    \geq \size{\unmap{\config2}} + 1
    \geq \size{\config2}
    .
  \end{equation}

  To see that (2) implies (1),
  suppose $\mf$ is free
  in a stable configuration $\config2$ of $\tbnadd$.
  We claim that $\unmap{\config2}$ is a stable configuration of $\tbnorig$ with $\ma$ and $\mb$ together. 
  To see that $\unmap{\config2}$ is stable,
  consider a saturated configuration $\config1$ of $\tbnorig$.
  We have
  \begin{equation}
    \size{\unmap{\config2}}
    \geq \size{\config2} - 1
    \geq \size{\map{\config1}} - 1
    = \size{\config1}
    .
  \end{equation}
  Now in $\config2$, if $\ma'$ and $\mb'$ were only held together by $\mfx$,
  then we could bind $\mfx$ instead to $\mf$
  to increase the number of polymers by one,
  which would contradict that $\config2$ is stable.
  So $\ma$ and $\mb$ are together in $\unmap{\config2}$.
  \qed
\end{proof}
}

\section{Conclusion} \label{conclusion}

The \tbn/ perspective complements kinetic models of chemical computation with a substrate-independent model of thermodynamic equilibrium based on counting bonds and configurational entropy.
We proved tight bounds on the computational complexity of $\pcount$,
the problem at the core of any questions about the behavior of TBNs,
as well as $\pfree$ and $\ptogether$,
problems motivated by applying the model to verify different chemical systems.
Although we prove these problems are worst-case intractable,
we developed algorithms and implementations
for solvers effective in many cases in practice.

An important future direction
is to consider the TBN model in the ``bulk'' experimental setting,
where monomer types are given in relative concentrations.
In this regime we need to consider larger complexes involving multiple copies of individual monomers types.
Note that the size of polymers in a stable configuration can be bounded as a function of the number of site and monomer types, but the bound is exponential~\cite{tbn}. 
Thus computationally, the bulk setting is more challenging.

Restricting TBNs in certain ways could be a useful design strategy for engineering systems that are powerful yet easy to verify.
There may be interesting classes of TBNs for which the problems we consider are provably easy.
Such classes might arise
from imposing a specific global or local property on the TBN.
More generally,
it is likely that even more sophisticated SAT reductions than ours
could effectively handle more complex problems of interest.

Understanding the nature of composition of TBN modules 
could achieve complex yet easy to verify behavior.
However, it is not clear in general how the stable configurations of $\therm1_1 \uplus \therm1_2$ relate to those of $\therm1_1$ and $\therm1_2$.
It appears hard to ``isolate'' distinct functional parts in TBNs because stability is a global property.

Making bonds stronger and decreasing the effective concentration
makes stable configurations arbitrarily preferred to unstable ones~\cite{tbn}.
However, in practice we may need to consider unstable configurations as well and limit which ones are undesirable.
If there are too many ``close to stable'' configurations, they may end up occurring with non-negligible probability.
One approach involves modifying TBNs to increase the ``distance to stability'' of undesirable configurations~\cite{tbn}. 
For example, in the case that monomer $\monomer1$ is not supposed to be free, constructions with a large $\maxpolys{\therm1} - \maxfreepolys{\therm1}$ have a large penalty to spuriously releasing $\monomer1$.
Our package can easily be extended to compute $\maxpolys{\therm1} - \maxfreepolys{\therm1}$ to verify the distance to stability.
Since the probability of undesired behavior depends also on the number of configurations, 
an important question is to count how many configurations with a given level of stability
have a given bad property.
How hard is this counting problem exactly?
For instance, is it in $\# \mathrm{P}$?

Computing with a \tbn/ as presented
involves only looking at the existence of certain stable configurations.
Instead, we can imagine unifying this thermodynamic equilibrium perspective
with a kinetic perspective entirely within the TBN model---by asking about which configurations can be reached without traversing thermodynamically unfavorable configurations~\cite{breikprogramming}.
The computational difficulty of this problem,
and what algorithms solve it,
remain to be answered.

\bigskip
\noindent\textbf{Acknowledgements.}
KB and DS were supported by NSF grants CCF-1618895 and CCF-1652824.
We thank Lakshmi Prakash for developing the original Mathematica interface to the SAT solver and for helpful discussions.
We also thank Cameron Chalk for help with the example in Figure~\ref{fig:circuit}.
We are also grateful to anonymous reviewers
for suggesting important improvements.

\bibliographystyle{plain}
\bibliography{SAT,chemcomp,complexity}

\begin{thebibliography}{10}

\bibitem{solver}
The accompanying solver is open source and is publicly available
  (\url{https://bitbucket.org/ksbtex/tbnsolverm/}). The high-level interface to
  the solver is via a Mathematica package. Please see the solver documentation
  for more information and examples.

\bibitem{audemard2009predicting}
Gilles Audemard and Laurent Simon.
\newblock Predicting learnt clauses quality in modern {SAT} solvers.
\newblock In {\em IJCAI}, volume~9, pages 399--404, 2009.

\bibitem{barish2009information}
Robert~D Barish, Rebecca Schulman, Paul~WK Rothemund, and Erik Winfree.
\newblock An information-bearing seed for nucleating algorithmic self-assembly.
\newblock {\em Proceedings of the National Academy of Sciences},
  106(15):6054--6059, 2009.

\bibitem{breikprogramming}
Keenan Breik, Cameron Chalk, David Doty, David Haley, and David Soloveichik.
\newblock Programming substrate-independent kinetic barriers with thermodynamic
  binding networks.
\newblock In {\em Proceedings of the 16th International Conference on
  Computational Methods in Systems Biology (CMSB)}, 2018.

\bibitem{chalk2018thermodynamically}
Cameron Chalk, Jacob Hendricks, Matthew~J Patitz, and Michael Sharp.
\newblock Thermodynamically favorable computation via tile self-assembly.
\newblock In {\em International Conference on Unconventional Computation and
  Natural Computation}, pages 16--31. Springer, 2018.

\bibitem{cherry2018scaling}
Kevin~M Cherry and Lulu Qian.
\newblock Scaling up molecular pattern recognition with {DNA}-based
  winner-take-all neural networks.
\newblock {\em Nature}, pages 10.1038/s41586--018--0289--6, 2018.

\bibitem{cipra2000ising}
Barry~A Cipra.
\newblock The ising model is {NP}-complete.
\newblock {\em SIAM News}, 33(6):1--3, 2000.

\bibitem{BMC}
Edmund~M. Clarke, Armin Biere, Richard Raimi, and Yunshan Zhu.
\newblock Bounded model checking using satisfiability solving.
\newblock {\em Formal Methods in System Design}, 19(1):7--34, 2001.

\bibitem{Copty2001}
Fady Copty, Limor Fix, Ranan Fraer, Enrico Giunchiglia, Gila Kamhi, Armando
  Tacchella, and Moshe~Y. Vardi.
\newblock Benefits of bounded model checking at an industrial setting.
\newblock In {\em CAV}, pages 436--453. Springer, 2001.

\bibitem{doty2012theory}
David Doty.
\newblock Theory of algorithmic self-assembly.
\newblock {\em Communications of the {ACM}}, 55(12):78--88, 2012.

\bibitem{tbn}
David Doty, Trent~A. Rogers, David Soloveichik, Chris Thachuk, and Damien
  Woods.
\newblock Thermodynamic binding networks.
\newblock In Robert Brijder and Lulu Qian, editors, {\em DNA Computing and
  Molecular Programming: 23rd International Conference}, pages 249--266.
  Springer, 2017.

\bibitem{BVE}
Niklas E{\'e}n and Armin Biere.
\newblock Effective preprocessing in {SAT} through variable and clause
  elimination.
\newblock In {\em SAT 2005}, volume 3569 of {\em LNCS}, pages 61--75. Springer,
  2005.

\bibitem{fal2009richer}
Piotr Faliszewski, Edith Hemaspaandra, Lane~A Hemaspaandra, and J{\"o}rg Rothe.
\newblock A richer understanding of the complexity of election systems.
\newblock In {\em Fundamental problems in computing}, pages 375--406. Springer,
  2009.

\bibitem{hart1997robust}
William~E Hart and Sorin Istrail.
\newblock Robust proofs of {NP}-hardness for protein folding: general lattices
  and energy potentials.
\newblock {\em Journal of Computational Biology}, 4(1):1--22, 1997.

\bibitem{HeuleS15}
Marijn Heule and Stefan Szeider.
\newblock A {SAT} approach to clique-width.
\newblock {\em {ACM} Trans. Comput. Log.}, 16(3):24:1--24:27, 2015.

\bibitem{SBMC}
Franjo Ivan{\v{c}}i{\'{c}}, Zijiang Yang, Malay~K. Ganai, Aarti Gupta, and
  Pranav Ashar.
\newblock Efficient {SAT}-based bounded model checking for software
  verification.
\newblock {\em Theoretical Computer Science}, 404(3):256--274, 2008.

\bibitem{jonoska2011stoichiometry}
Natasha Jonoska, Gregory~L McColm, and Ana Staninska.
\newblock On stoichiometry for the assembly of flexible tile {DNA} complexes.
\newblock {\em Natural Computing}, 10(3):1121--1141, 2011.

\bibitem{kadin1989pnp}
Jim Kadin.
\newblock {$\text{{P}}^{\text{NP}[O(\log n)]}$} and sparse {T}uring-complete
  sets for {NP}.
\newblock {\em Journal of Computer and System Sciences}, 39(3):282--298, 1989.

\bibitem{lyngso2000rna}
Rune~B Lyngs{\o} and Christian~NS Pedersen.
\newblock {RNA} pseudoknot prediction in energy-based models.
\newblock {\em Journal of computational biology}, 7(3-4):409--427, 2000.

\bibitem{BVA}
Norbert Manthey, Marijn J.~H. Heule, and Armin Biere.
\newblock Automated reencoding of {B}oolean formulas.
\newblock In {\em Proceedings of Haifa Verification Conference 2012}, 2012.

\bibitem{ong2017programmable}
Luvena~L Ong, Nikita Hanikel, Omar~K Yaghi, Casey Grun, Maximilian~T Strauss,
  Patrick Bron, Josephine Lai-Kee-Him, Florian Schueder, Bei Wang, Pengfei
  Wang, et~al.
\newblock Programmable self-assembly of three-dimensional nanostructures from
  10,000 unique components.
\newblock {\em Nature}, 552(7683):72, 2017.

\bibitem{phillips2009programming}
Andrew Phillips and Luca Cardelli.
\newblock A programming language for composable {DNA} circuits.
\newblock {\em Journal of the Royal Society Interface}, 6(Suppl 4):S419--S436,
  2009.

\bibitem{SchWin09}
Rebecca Schulman and Erik Winfree.
\newblock Programmable control of nucleation for algorithmic self-assembly.
\newblock {\em SIAM Journal on Computing}, 39(4):1581--1616, 2009.

\bibitem{Sinz}
Carsten Sinz.
\newblock Towards an optimal {CNF} encoding of {B}oolean cardinality
  constraints.
\newblock In {\em Principles and Practice of Constraint Programming}, pages
  827--831, 2005.

\bibitem{Spa2005}
Holger Spakowski.
\newblock {\em Completeness for parallel access to NP and counting class
  separations.}
\newblock PhD thesis, 2005.

\bibitem{srinivas2017enzyme}
Niranjan Srinivas, James Parkin, Georg Seelig, Erik Winfree, and David
  Soloveichik.
\newblock Enzyme-free nucleic acid dynamical systems.
\newblock {\em Science}, 358, 2017.

\bibitem{thachuk2015leakless}
Chris Thachuk, Erik Winfree, and David Soloveichik.
\newblock Leakless {DNA} strand displacement systems.
\newblock In {\em Proceedings of {DNA} Computing and Molecular Programming 21},
  pages 133--153. Springer, 2015.

\bibitem{zhang2011dynamic}
David~Yu Zhang and Georg Seelig.
\newblock Dynamic {DNA} nanotechnology using strand-displacement reactions.
\newblock {\em Nature chemistry}, 3(2):103--113, 2011.

\end{thebibliography}

\end{document}